\documentclass[conference]{IEEEtran}
\IEEEoverridecommandlockouts
\usepackage{cite}
\usepackage{amsmath,amssymb,amsfonts}

\usepackage{graphicx}
\usepackage{booktabs}
\usepackage{array}
\usepackage{graphicx} 
\usepackage{float}
\usepackage{comment}
\usepackage{xurl}                     
\usepackage[breaklinks]{hyperref}    
\usepackage{mathtools}               
\allowdisplaybreaks                   
\usepackage{adjustbox}                
\usepackage{graphicx}                 
\usepackage{longtable}                
\usepackage{url} 

\usepackage{textcomp}
\usepackage{xcolor}
\usepackage[table]{xcolor} 
\usepackage{booktabs}     
\usepackage{multirow}      

\usepackage{booktabs}
\usepackage{multirow}
\usepackage{amsthm}

\newtheorem{theorem}{Theorem}

\usepackage{booktabs}
\usepackage{multirow}
\usepackage[T1]{fontenc}
\usepackage{amssymb}

\usepackage{amsmath}
\usepackage{longtable}
\usepackage[utf8]{inputenc}
\usepackage[T1]{fontenc}
\usepackage{xcolor}
\usepackage{amsmath,amssymb,amsfonts}
\usepackage{textcomp}
\usepackage{float}
\usepackage{array}
\usepackage{booktabs} 
\usepackage{bm}
\usepackage{epsfig}
\usepackage{soul}
\usepackage{pgf}
\usepackage{multirow}
\usepackage{bigstrut}
\usepackage{comment}
\usepackage{rotating}
\usepackage{pdflscape}
\usepackage{amsmath} 
\usepackage{lscape}
\usepackage{subcaption}
\usepackage{setspace}
\usepackage{amssymb}
\usepackage{amsmath}
\usepackage{relsize}
\usepackage{xcolor}
\usepackage{url}
\usepackage{algorithm}
\usepackage{algpseudocode}
\usepackage{graphicx}
\usepackage{multirow}
\usepackage{multicol}
\usepackage{comment}
\usepackage{booktabs} 
\usepackage{array} 
\usepackage{balance} 
\usepackage{soul}
\usepackage{siunitx}
\usepackage{graphicx}
\usepackage{subcaption} 
\usepackage{float}
\usepackage{xcolor}
\usepackage{comment}
\usepackage{graphicx}
\usepackage{tikz}
\usetikzlibrary{positioning, shapes, arrows.meta}
\usepackage{array}
\usepackage{xcolor}
\usepackage{pifont}
\usepackage{tikz}
\usetikzlibrary{arrows.meta,shapes,positioning} 

\def\BibTeX{{\rm B\kern-.05em{\sc i\kern-.025em b}\kern-.08em
    T\kern-.1667em\lower.7ex\hbox{E}\kern-.125emX}}
\begin{document}

\title{Traveling Salesman Problem with a preprocessing method for classical and quantum optimization}

\author{
\IEEEauthorblockN{Alessia Ciacco\IEEEauthorrefmark{1},
Luigi Di Puglia Pugliese\IEEEauthorrefmark{2},
Francesca Guerriero\IEEEauthorrefmark{1}}

\IEEEauthorblockA{\IEEEauthorrefmark{1}Department of Mechanical, Energy and Management Engineering,
University of Calabria, Rende (CS), Italy}

\IEEEauthorblockA{\IEEEauthorrefmark{2}Istituto di Calcolo e Reti ad Alte Prestazioni,
Consiglio Nazionale delle Ricerche (ICAR-CNR),
Rende (CS), Italy}


\IEEEauthorblockA{Emails: alessia.ciacco@unical.it,
luigi.dipugliapugliese@icar.cnr.it,
francesca.guerriero@unical.it}
}

\maketitle

\begin{abstract}
The Traveling Salesman Problem is a fundamental combinatorial optimization problem widely studied in operations research. Despite its simple formulation, it remains computationally challenging due to the exponential growth of the search space and the large number of constraints required to eliminate subtours. This paper introduces a preprocessing strategy that significantly reduces the size of the optimization model by restricting the set of candidate arcs and retaining only the lowest-cost neighbors for each vertex. Computational experiments on TSPLIB benchmark instances demonstrate that the proposed approach substantially reduces the number of decision variables. The method is evaluated using both classical and quantum optimization techniques, showing improvements in computational time and { consistent reductions in optimality gaps, while maintaining high-quality solutions}. Overall, the results indicate that the proposed preprocessing enhances the scalability of the formulation and makes it more suitable for both classical solvers and emerging quantum optimization frameworks.
\end{abstract}


\section{Introduction}

The Traveling Salesman Problem (TSP) is one of the most extensively studied combinatorial optimization problems in operations research. 
 Given a set of cities and pairwise travel costs between them, the objective is to determine the shortest possible route that visits each city exactly once and returns to the starting point. 

From a computational complexity perspective, the problem has been proven to be NP-hard \cite{karp1972reducibility}. This classification implies that no polynomial-time algorithm is known for solving the problem exactly in the general case. Consequently, research on the TSP has focused on the development of both exact methods capable of solving moderately large instances and heuristic or metaheuristic approaches that provide high-quality approximate solutions for large-scale problems.

In recent years, the emergence of quantum computing has opened new perspectives for addressing computationally challenging combinatorial optimization problems. In particular, hybrid quantum-classical optimization approaches have attracted significant attention as potential tools for exploring large combinatorial search spaces. Among the available paradigms, Quantum Annealing (QA) is a metaheuristic specifically designed for optimization tasks \cite{rajak2023quantum, hauke2020perspectives}. QA is based on the gradual evolution of a quantum system toward the ground state of a problem Hamiltonian encoding the objective function. QA has been investigated for a variety of combinatorial optimization problems. Applications include routing problems~\cite{ciacco2025steiner, holliday2025advanced, osaba2025quantum, ciacco2025steiner2}, facility location~\cite{ciacco2026facility, malviya2023logistics}, packing problems~\cite{de2022hybrid, cellini2024qal}, personnel scheduling \cite{li2023quantum, ciacco2025Educational}. A broader overview of quantum approaches is provided in the survey by \cite{ciacco2025review}.
Despite these advances, applying QA to routing problems such as the TSP remains challenging. One of the early attempts to solve the TSP on quantum annealers is presented by \cite{jain2021tsp}, where the problem is reformulated as a Quadratic Unconstrained Binary Optimization (QUBO) model and implemented on a D-Wave processor. Experimental results show important limitations: due to embedding requirements and hardware connectivity constraints, the approach can only handle very small instances.
\cite{warren2020tsp} analyzes several software implementations for solving the TSP on D-Wave quantum annealers. The study shows that current hardware cannot directly handle realistic TSP instances. \cite{bochkarev2026quantum} observe that combinatorial problems such as the TSP require complex constraint structures that significantly increase the number of variables and interactions in QUBO models. These limitations highlight the need for preprocessing techniques capable of reducing the model size and restricting the set of candidate arcs before the optimization stage.

In this paper, { we focus on the Euclidian TSP, hence we address the symmetric version of the problem.} We propose a novel preprocessing strategy aimed at reducing the size of the TSP formulation while preserving the feasibility of the problem. The method restricts the set of candidate arcs considered in the optimization model by retaining, for each vertex, only a subset of neighboring vertices associated with the lowest travel costs. To evaluate the effectiveness of the preprocessing strategy, we solve the resulting TSP formulations using both classical and quantum optimization techniques. 
From the quantum optimization perspective, the proposed approach enables the {treatment} of larger TSP instances than those typically considered in the current quantum optimization literature. In particular, thanks to the reduction in the number of candidate arcs introduced by the preprocessing step, we are able to {handle} TSP instances with up to 15 customers using a quantum optimization approach. 
{The results indicate an improvement in scalability compared to existing quantum approaches. In this context, the aim of this work is not to demonstrate the superiority of quantum optimization over classical methods, but rather to show how the proposed preprocessing strategy reduces the model size and enables current hybrid quantum solvers to handle slightly larger TSP instances than those typically considered in the literature.}
Moreover, the computational experiments are conducted on benchmark instances commonly used in the literature rather than on artificially generated toy problems. This provides a more realistic evaluation of the potential of quantum optimization methods for routing problems and represents one of the first attempts to address quantum TSP using established benchmark datasets.

The remainder of the paper is organized as follows. Section~\ref{sec:problem_statement} presents the formal definition of the TSP and its mathematical formulation. Section~\ref{sec:caf} introduces the proposed preprocessing method and discusses its theoretical properties and feasibility guarantees. Section~\ref{sec:solution_approach} describes the overall solution approach adopted to solve the problem using both classical and quantum optimization techniques. Section~\ref{sec:experiments} reports the computational experiments, including a thorough analysis of the results obtained.
Finally, Section~\ref{sec:conclusions} concludes the paper and outlines possible directions for future research.

\section{Problem definition}\label{sec:problem_statement}

The TSP consists of finding a minimum-cost tour that visits each city exactly once and returns to the starting city. Let $G = (V, A)$ be a complete graph, where $V = \{1, \dots, n\}$ denotes the set of cities and $A$ the set of arcs. 
Let $S$ be any proper subset of the vertex set $V$, i.e., $S \subset V$, representing a subset of cities. Each arc $(i,j) \in A$ is associated with a non-negative cost $c_{ij}$, representing the distance cost between city $i$ and city $j$. { We assume that} the cost matrix satisfies $c_{ij} = c_{ji}, \; \forall\: i,j \in V$, and self-loops are excluded, i.e., arcs of the form $(i,i)$ are not included in $A$. 
The objective of the TSP is to determine a minimum-cost Hamiltonian cycle on $G$.
The parameters used in the formulation of the TSP are summarized in Table~\ref{tab:notations}.

\begin{table}[ht]
\centering
\begin{tabular}{>{\centering\arraybackslash}m{1cm}>{\arraybackslash}m{6cm}}
\toprule
\textbf{Notation}  &   \textbf{Description} \\
\midrule
$V$ & set of all cities \\
$A$  & set of arcs in the graph \\
$S$  & any proper subset of set $V$ \\
\hline
$c_{ij}$ & distance cost between node $i$ and node $j$\\
\bottomrule
\end{tabular}
\caption{Notation used in the problem formulation.} \label{tab:notations}
\vspace{-0.5cm}
\end{table}

\subsection{Mathematical Formulation}\label{sec:mathematical_formulation}
The variable used to define the model is:
\begin{description}
    \item[$x_{ij}$] $\begin{cases} 
    1, & \text{if arc } (i,j) \in A \text{ is included in the tour}  \\
    0, & \text{otherwise}
    \end{cases}$
\end{description}

The TSP can be formulated as an Integer Linear Programming (ILP) model as follows:
\begin{align}
    &\text{min} \quad \sum_{(i,j)\in A} c_{ij} x_{ij} \label{eq:objective} \\
    &\text{s.t.} \nonumber \\
    &\sum_{i \in V : (i,j)\in A} x_{ij} = 1, \quad \forall j \in V \label{eq:visit_once1} \\
    &\sum_{i \in V : (j,i)\in A} x_{ji} = 1, \quad \forall j \in V \label{eq:visit_once2} \\
    &\sum_{\substack{i \in S}} \sum_{\substack{j \in  S}} x_{ij} \leq |S| - 1 , \quad  S\subset V, |S|\geq 2\label{eq:subtour} \\  
    &x_{ij} \in \{0,1\}, \quad \forall (i,j) \in A \label{eq:binary}
\end{align}

The objective function~\eqref{eq:objective} minimizes the distance cost. 
Constraints~\eqref{eq:visit_once1}-\eqref{eq:visit_once2} impose the degree conditions at each vertex. For every node ($j \in V$), exactly one incoming arc and exactly one outgoing arc are selected in the solution. This ensures that each city is visited exactly once and that the resulting solution consists of a collection of directed cycles covering all vertices.
Constraints~\eqref{eq:subtour} are subtour elimination constraints whose purpose is to prevent the formation of disconnected cycles in the solution. 
Constraints~\eqref{eq:binary} define the domain of the decision variables.

\section{Cost-based Arc Filtering (CAF)} \label{sec:caf}
We introduce the CAF preprocessing method to reduce the size of the optimization model by restricting the set of candidate arcs and retaining, for each vertex, only the neighbors associated with the lowest travel costs. This reduction is motivated by the fact that, in a complete graph, the number of arcs is $|A| = n(n - 1)$, which grows quadratically with the number of vertices and results in a substantial number of decision variables for large instances. Leveraging structural properties of Hamiltonian graphs, the proposed filtering procedure reduces the model size while preserving the existence of Hamiltonian cycles.

Our approach builds on Dirac’s theorem \cite{dirac1952some}, a classical result in graph theory. Consider an undirected graph $G = (V, E)$ with $n = |V| \geq 3$ vertices, and let $\delta(G)$ denote its minimum degree. Dirac showed that if $\delta(G) \geq \frac{n}{2}$, then $G$ is Hamiltonian, i.e., it contains a cycle visiting each vertex exactly once. This result implies that a graph may remain Hamiltonian even after the removal of a large number of edges, provided that each vertex maintains a sufficiently high degree.

Motivated by this observation, the CAF procedure constructs a reduced set of candidate arcs by retaining, for each vertex, only those with the smallest associated costs. 
For each vertex $i \in V$, the arcs $(i,j)$ with $j \neq i$ are sorted in nondecreasing order of cost $c_{ij}$, and only the first $k = \left\lceil \frac{n}{2} \right\rceil$ neighbors are retained.

The procedure is summarized in Algorithm~\ref{afgr}.

\begin{algorithm}[ht]
\caption{CAF}
\label{afgr}
\begin{algorithmic}[1]
\State \textbf{Input:} $V$: set of vertices, $c_{ij}$: travel cost matrix
\State \textbf{Output:} $\tilde{A}$: reduced set of candidate arcs

\State $n \gets |V|$
\State $k \gets \lceil n/2 \rceil$
\State $\tilde{A} \gets \emptyset$

\For{each $i \in V$}
    \State Order vertices $j \neq i$ according to increasing cost $c_{ij}$
    \State Select the first $k$ vertices in the ordered list
    \For{each selected $j$}
        \State Add arcs $(i,j)$ and $(j,i)$ to $\tilde{A}$
    \EndFor
\EndFor

\State \Return $\tilde{A}$
\end{algorithmic}

\end{algorithm}

\subsection{Feasibility Preservation}

A potential drawback of arc filtering procedures is that removing arcs may destroy all Hamiltonian cycles, rendering the instance infeasible. 
We show that choosing $k = \lceil n/2 \rceil$ guarantees that the reduced graph still contains at least one Hamiltonian cycle.

\begin{theorem}[Feasibility of the Reduced Graph]
Let $G=(V,A)$ be the complete directed graph and let $\tilde{G}=(V,\tilde{A})$ be the directed graph obtained by applying the CAF procedure with $k=\lceil n/2\rceil$. 
Then the undirected graph induced by $\tilde{A}$ is Hamiltonian.
\end{theorem}

\begin{proof}
Let $G_U=(V,A)$ denote the undirected graph obtained from $\tilde{G}$ by replacing each pair of symmetric arcs $(i,j)$ and $(j,i)$ with the undirected arc $\{i,j\}$.
By construction of the CAF procedure, for each vertex $i \in V$ the algorithm selects $k$ distinct neighbors with the smallest travel costs. 
Therefore each vertex $i$ is adjacent to at least $k$ distinct vertices in $G_U$, and the minimum degree satisfies
$$
\delta(G_U) \ge k = \left\lceil \frac{n}{2} \right\rceil \ge \frac{n}{2}.
$$
Since $n \ge 3$, Dirac's theorem implies that $G_U$ is Hamiltonian. 
Consequently, the reduced graph produced by the CAF procedure always admits at least one Hamiltonian cycle.
\end{proof}

\subsection{Optimality Considerations}

The CAF procedure is designed to reduce the number of candidate arcs while preserving a sufficiently dense graph structure that guarantees the existence of Hamiltonian cycles. The filtering step does not impose explicit conditions aimed at preserving the optimal tour of the original complete graph. Instead, its design is motivated by structural properties commonly observed in Euclidean TSP instances.
In Euclidean TSP instances, the travel costs correspond to Euclidean distances between points in $\mathbb{R}^2$, and therefore satisfy the triangle inequality
$$
c_{ij} \le c_{ik} + c_{kj}, \quad \forall i,j,k \in V.
$$
These instances exhibit well-known geometric characteristics that influence the structure of optimal tours. In particular, optimal tours tend to be composed primarily of relatively short arcs connecting nearby vertices, while long-distance connections rarely appear in optimal solutions \cite{reinelt2001traveling, applegate2011traveling}.

The CAF procedure leverages this structural behavior by retaining, for each vertex $i$, the $k=\lceil n/2\rceil$ incident arcs with the smallest travel costs. Since these correspond to the closest vertices in Euclidean space, the filtering process removes primarily long-distance connections while preserving a dense set of geometrically local arcs.

The resulting reduced arc set $\tilde{A}$ therefore concentrates the candidate arcs in the region of the solution space where optimal or near-optimal tours are typically found, while significantly decreasing the size of the optimization model. {However, since the filtering procedure does not explicitly enforce the preservation of the optimal tour, optimality is not guaranteed, and deviations from the global optimum may occur.}

\section{Solution approach} \label{sec:solution_approach}

We adopt a structured multi-phase approach to solve the TSP by combining a preprocessing strategy with both classical and quantum optimization techniques. The process begins with the formulation of the problem as an ILP model. To improve the tractability of the problem, we introduce the CAF approach, which reduces the set of candidate arcs while preserving the feasibility of the tour.
To assess the effectiveness of the preprocessing step, the ILP formulation is solved both with and without CAF using a classical commercial solver. Subsequently, the reduced problem is also addressed within a quantum optimization framework using D-Wave's hybrid solver. The overall solution process is structured as follows:

\begin{enumerate}

\item \textbf{Mathematical formulation}: the TSP is formulated as an ILP model representing the standard mathematical formulation of the routing problem.

\item \textbf{CAF preprocessing}: the CAF approach is applied to reduce the number of candidate arcs by retaining, for each node, only a subset of neighboring nodes associated with the lowest travel costs. This preprocessing step reduces the size of the optimization model while preserving a graph structure that still admits Hamiltonian tours.

\item \textbf{Classical optimization analysis}: the ILP formulation is solved with and without CAF using the Gurobi solver in order to evaluate the impact of the preprocessing strategy on the computational tractability of the model.

\item \textbf{Quantum optimization}: the preprocessed formulation is then solved within a quantum optimization framework implemented through D-Wave's \texttt{LeapCQMHybrid} solver. The problem is formulated within the Constrained Quadratic Model (CQM) framework, which supports the integration of linear constraints together with binary and integer decision variables within a unified optimization model.

\end{enumerate}

\section{Computational study} \label{sec:experiments} 
All experiments are performed on a Windows~10 workstation equipped with an Intel processor (4 physical cores, 8 threads) and 15.6~GB of RAM, using Python~3 within a Jupyter Notebook environment. Exact optimization experiments are conducted using Gurobi Optimizer~11.0.1. The CQM formulation is implemented using D-Wave’s Ocean SDK, which provides the Python interface for modeling and executing hybrid quantum–classical optimization workflows. In particular, the libraries \texttt{dimod}~0.12.18 and \texttt{dwave-system}~1.28.0 are used to define the CQM model and submit quantum queries to the D-Wave platform. These queries are executed on the \texttt{Advantage2\_system1.4} quantum processing unit (4596 qubits, Zephyr topology). Unless otherwise specified, solver parameters are kept at their default values to ensure consistency and reproducibility across experiments. {In particular, the \texttt{LeapHybridCQMSampler} solver provided by D-Wave is used to solve the CQM model, with a time limit of 10 seconds per run specified through the \texttt{time\_limit} parameter, while all other solver parameters are left at their default settings.} Further details regarding the solver and its applications can be found in \cite{benson2023cqm, osaba2024solving, osaba2025d}.

All computational experiments are conducted on benchmark instances derived from the \texttt{berlin52.tsp} dataset from the TSPLIB benchmark library \cite{reinelt1991tsplib}, one of the most widely used collections of standard instances for evaluating algorithms for the TSP. TSPLIB provides well-established benchmarks with known optimal solutions, enabling rigorous and reproducible evaluation of optimization approaches. The use of these standardized instances facilitates meaningful comparisons with previously published results in the TSP literature.

The experimental evaluation is organized into three main parts. The first part evaluates the impact of the CAF preprocessing on the size of the optimization model, highlighting its effect on the reduction of decision variables and overall model complexity. The second part presents the results obtained using classical optimization methods, analyzing the performance of the proposed approach in terms of solution quality and computational efficiency. {It is important to note that the goal of the classical experiments is not to compete with large-scale specialized TSP solvers, but rather to evaluate the relative impact of the CAF preprocessing on model size and computational effort within a controlled setting.} Finally, the third part reports the experiments conducted using the hybrid quantum solver, assessing the effectiveness of the CAF preprocessing in improving the quality of the solutions obtained within the quantum optimization framework.
\subsection{Impact of CAF on the model size}
Table~\ref{tab:milp_reduction} highlights the impact of the CAF preprocessing on the size of the optimization model. 
\begin{table}[ht]
\centering
\resizebox{\columnwidth}{!}{
\begin{tabular}{c|cc|c || c|cc|c}
\toprule
\multicolumn{8}{c}{\textbf{Number of variables }} \\
\midrule

$V$ & Without CAF & With CAF & GAP &
$V$ & Without CAF & With CAF & GAP \\

\midrule    
   5     & 20    & 18    & 10\%  & 28    & 756   & 486   & 36\% \\
    6     & 30    & 24    & 20\%  & 29    & 812   & 540   & 33\% \\
    7     & 42    & 34    & 19\%  & 30    & 870   & 562   & 35\% \\
    8     & 56    & 42    & 25\%  & 31    & 930   & 622   & 33\% \\
    9     & 72    & 58    & 19\%  & 32    & 992   & 648   & 35\% \\
    10    & 90    & 64    & 29\%  & 33    & 1056  & 708   & 33\% \\
    11    & 110   & 88    & 20\%  & 34    & 1122  & 738   & 34\% \\
    12    & 132   & 96    & 27\%  & 35    & 1190  & 806   & 32\% \\
    13    & 156   & 118   & 24\%  & 36    & 1260  & 838   & 33\% \\
    14    & 182   & 126   & 31\%  & 37    & 1332  & 910   & 32\% \\
    15    & 210   & 154   & 27\%  & 38    & 1406  & 946   & 33\% \\
    16    & 240   & 166   & 31\%  & 39    & 1482  & 1022  & 31\% \\
    17    & 272   & 194   & 29\%  & 40    & 1560  & 1056  & 32\% \\
    18    & 306   & 208   & 32\%  & 41    & 1640  & 1136  & 31\% \\
    19    & 342   & 244   & 29\%  & 42    & 1722  & 1166  & 32\% \\
    20    & 380   & 262   & 31\%  & 43    & 1806  & 1238  & 31\% \\
    21    & 420   & 294   & 30\%  & 44    & 1892  & 1274  & 33\% \\
    22    & 462   & 312   & 32\%  & 45    & 1980  & 1358  & 31\% \\
    23    & 506   & 352   & 30\%  & 46    & 2070  & 1390  & 33\% \\
    24    & 552   & 376   & 32\%  & 47    & 2162  & 1472  & 32\% \\
    25    & 600   & 420   & 30\%  & 48    & 2256  & 1514  & 33\% \\
    26    & 650   & 432   & 34\%  & 49    & 2352  & 1608  & 32\% \\
    27    & 702   & 474   & 32\%  & 50    & 2450  & 1652  & 33\% \\

\bottomrule
\end{tabular}
}
\caption{Comparison of the size of the TSP ILP formulation without and with the CAF preprocessing procedure. Each row corresponds to a different instance size defined by the number of vertices $V$. The table reports:  
\textbf{Without CAF}: number of decision variables in the complete ILP formulation, 
\textbf{With CAF}: number of decision variables after applying the CAF preprocessing procedure, 
\textbf{GAP}: percentage reduction in the number of variables obtained by CAF with respect to the complete formulation.}
\label{tab:milp_reduction}
\end{table}
The results highlight that the CAF preprocessing procedure reduces the number of decision variables by approximately 30\% with respect to the original formulation. This reduction {significantly reduces the number of decision variables and leads to a more compact optimization model.}
As the number of vertices increases, the difference between the original formulation and the reduced formulation becomes more evident, confirming that the impact of the preprocessing grows with the problem size.
Figure~\ref{fig:variables_comparison} illustrates this behavior by comparing the number of decision variables in the formulation without CAF and with CAF as the number of vertices increases. The two curves show how the number of variables grows for the two formulations as 
$V$ increases. The curve associated with the formulation with CAF remains consistently below the one without CAF. In particular, the gap between the two curves gradually widens as $V$ increases, highlighting the effectiveness of the preprocessing in reducing the model size and improving scalability.
\begin{figure}[h]
\centering
\includegraphics[width=0.9\linewidth]{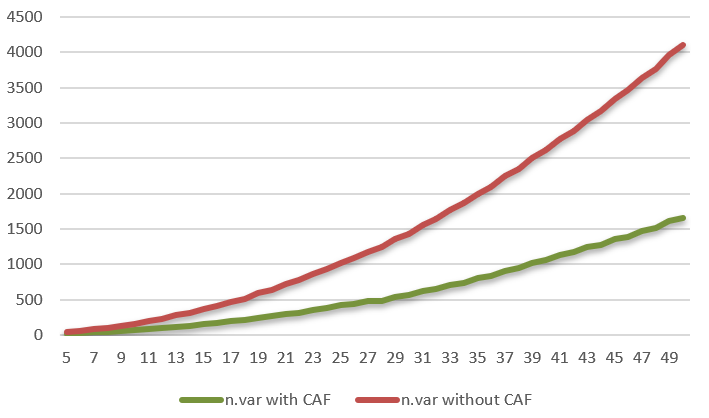}
\caption{Number of decision variables in the TSP as the number of vertices $V$ increases. The red curve represents the number of variables in the standard formulation (n.var without CAF), while the green curve shows the number of variables after applying the CAF preprocessing procedure (n.var with CAF).}
\label{fig:variables_comparison}
\vspace{-0.5cm}
\end{figure}

\subsection{Classical optimization experiments}
Table~\ref{tab:milp_results} reports the computational performance of the complete ILP formulation of the TSP for instances with an increasing number of vertices $V$, both without and with the CAF preprocessing procedure. In all experiments, the solver is executed with a time limit of 300 seconds. 

\begin{table}[h]
\centering
\resizebox{\columnwidth}{!}{
\begin{tabular}{c|cccc|cccc|c}
\toprule
\multirow{2}{*}{\textbf{V}} & \multicolumn{4}{c|}{\textbf{Without CAF}} & \multicolumn{4}{c|}{\textbf{With CAF}} & \multirow{2}{*}{\textbf{CAF gap}} \\
\cmidrule(lr){2-5} \cmidrule(lr){6-9}
 & \textbf{OF} & \textbf{Time}& \textbf{Time Solve} & \textbf{GAP Opt} & \textbf{OF} & \textbf{Time}& \textbf{Time Solve}  & \textbf{GAP Opt} & \\
\midrule
    5     & 2314.55 & 0.00  & 0.00  & 0\%   & 2314.55 & 0.08  & 0.04  & 0\% & 0.0\% \\
    6     & 2315.15 & 0.00  & 0.00  & 0\%   & 2323.20 & 0.00  & 0.00  & 0\% & 0.4\%\\
    7     & 2321.39 & 0.01  & 0.00  & 0\%   & 2321.39 & 0.01  & 0.00  & 0\% & 0.0\% \\
    8     & 2550.94 & 0.01  & 0.00  & 0\%   & 2550.94 & 0.02  & 0.01  & 0\% & 0.0\% \\
    9     & 2820.38 & 0.04  & 0.02  & 0\%   & 2874.44 & 0.02  & 0.01  & 0\% & 1.9\%\\
    10    & 2826.50 & 0.09  & 0.03  & 0\%   & 2826.50 & 0.08  & 0.04  & 0\% & 0.0\% \\
    11    & 4038.44 & 0.21  & 0.09  & 0\%   & 4038.44 & 0.16  & 0.08  & 0\% & 0.0\% \\
    12    & 4056.68 & 0.50  & 0.17  & 0\%   & 4056.68 & 0.30  & 0.17  & 0\% & 0.0\% \\
    13    & 4564.46 & 1.25  & 0.47  & 0\%   & 4564.46 & 0.81  & 0.49  & 0\% & 0.0\% \\
    14    & 4946.85 & 3.26  & 1.32  & 0\%   & 4965.33 & 1.65  & 1.15  & 0\% & 0.4\% \\
    15    & 4967.30 & 8.20  & 2.97  & 0\%   & 4967.30 & 4.70  & 3.28  & 0\%& 0.0\%  \\
    16    & 4990.46 & 20.04 & 7.80  & 0\%   & 4990.46 & 12.04 & 9.32  & 0\%& 0.0\%  \\
    17    & 5048.45 & 48.11 & 20.21 & 0\%   & 5048.45 & 31.17 & 25.08 & 0\% & 0.0\% \\
    18    & 5139.38 & 107.59 & 45.38 & 0\%   & 5139.38 & 57.57 & 47.01 & 0\% & 0.0\% \\
    19    & 5164.22 & 257.61 & 109.74 & 0\%   & 5164.22 & 139.12 & 115.51 & 0\% & 0.0\% \\
    20    & 5270.86 & 670.99 & 264.14 & 0\%   & 5270.86 & 316.13 & 257.66 & 0\% & 0.0\% \\
    21    & 12779.84 & 1239.27 & 306.43 & --   & 10683.42 & 437.07 & 305.67 & --& -16.4\% \\
   
\bottomrule
\end{tabular}
}
\caption{Computational results for the complete ILP formulation on instances with an increasing number of vertices $V$, comparing the model performance without and with the CAF preprocessing procedure. Each row corresponds to a different instance size defined by the number of vertices $V$. The table reports: 
\textbf{OF}: objective function value returned by the solver, 
\textbf{Time}: total computational time in seconds required to run the model (including instance construction and solution time), 
\textbf{Time Solve}: time in seconds required by the solver to solve the optimization model, 
\textbf{GAP Opt}: percentage optimality gap with respect {to the best solution found} by the solver.  \textbf{CAF gap}: percentage variation between the objective values obtained with and without CAF preprocessing.}
\label{tab:milp_results}
\vspace{-0.3cm}
\end{table}

For instances up to $V = 20$, the solver is able to prove optimality in both configurations, as indicated by the zero optimality gap. {Importantly, the CAF preprocessing preserves feasibility, although slight differences in the objective value may arise, as some arcs of the original optimal tour may be excluded during filtering. Consequently, the solutions obtained with CAF are optimal for the reduced formulation and remain close to those of the original problem.} The main benefit of CAF emerges in terms of computational time. While the difference is negligible for small instances, for medium and larger instances the preprocessing significantly reduces the solving time. For example, at $V=18$ the total time decreases from $107,59$ s to $57,57$ s, and at $V=19$ from $257,61$ s to $139,12$ s. A similar effect is observed for $V=20$, where the computational effort is reduced by more than half. Overall, CAF reduces the computational time by approximately $32\%$ on average, with reductions between $40\%$ and $54\%$ for medium and large instances ($V \geq 10$). For the instance with $V=21$, the solver reaches the imposed time limit of $300$ seconds without computing a valid lower bound, resulting in an infinite optimality gap.
{ The CAF gap is $0\%$ for most instances. In the remaining cases, it stays below $0.5\%$ twice ($V =4 \text{ and } 14$) and below $2\%$ once ($V =9$), except for $V=21$, where CAF improves the objective value by $16.4\%$.}
This occurs because the time limit applies only to the optimization phase and the large number of subtour constraints makes the root node extremely large. Nevertheless, the solution obtained with CAF achieves a smaller objective value than the one obtained without preprocessing, highlighting the practical advantage of the proposed reduction even when optimality cannot be certified within the time limit. 

\subsection{Quantum hybrid solver experiments}

Table~\ref{tab:performance-comparison1} presents the computational results obtained from the execution of the proposed model, comparing the performance with and without the CAF preprocessing method. For each instance, five independent runs were performed. 

\begin{table*}[ht]
\centering
\resizebox{\textwidth}{!}{
\begin{tabular}{cc|cccc|cccc|cccc|cccc}
\toprule
\multirow{2}{*}{$V$} & \multirow{2}{*}{OF opt}
& \multicolumn{8}{c|}{\textbf{Without CAF}}
& \multicolumn{8}{c}{\textbf{With CAF}} \\
\cmidrule(lr){3-10} \cmidrule(lr){11-18}
& & OF Avg & OF Std & GAP & \%Solved & Time Avg & Time Std & Time Solve Avg & Time Solve Std
& OF Avg & OF Std & GAP & \%Solved & Time Avg & Time Std & Time Solve Avg & Time Solve Std\\
\midrule
5  & 2314.55 & 2314.55 & 0.00 & 0\%  & 100\% & 14.91 & 0.08 & 10.21 & 0.04 & 2314.55 & 0.00 & 0\%  & 100\% & 15.63 & 1.32 & 10.21 & 0.01 \\
6  & 2315.15 & 2315.15 & 0.00 & 0\%  & 100\% & 15.01 & 0.29 & 10.22 & 0.03 & 2323.20 & 0.00 & 0\%  & 100\% & 14.91 & 0.10 & 10.21 & 0.01 \\
7  & 2321.39 & 2321.39 & 0.00 & 0\%  & 100\% & 14.97 & 0.08 & 10.29 & 0.04 & 2321.39 & 0.00 & 0\%  & 100\% & 15.05 & 0.26 & 10.26 & 0.05 \\
8  & 2550.94 & 2550.94 & 0.00 & 0\%  & 100\% & 15.12 & 0.10 & 10.24 & 0.04 & 2550.94 & 0.00 & 0\%  & 100\% & 15.56 & 1.16 & 10.26 & 0.04 \\
9  & 2820.38 & 2820.38 & 0.00 & 0\%  & 100\% & 15.37 & 0.04 & 10.23 & 0.01 & 2874.44 & 0.00 & 0\%  & 100\% & 15.43 & 0.26 & 10.22 & 0.01 \\
10 & 2826.50 & 2826.50 & 0.00 & 0\%  & 100\% & 15.61 & 0.07 & 10.30 & 0.05 & 2826.50 & 0.00 & 0\%  & 100\% & 15.63 & 0.32 & 10.23 & 0.04 \\
11 & 4038.44 & 4038.44 & 0.00 & 0\%  & 100\% & 16.38 & 0.55 & 10.29 & 0.05 & 4041.13 & 6.01 & 0\%  & 100\% & 15.58 & 0.16 & 10.29 & 0.07 \\
12 & 4056.68 & 4621.30 & 222.30 & 14\% & 100\% & 17.15 & 0.91 & 10.42 & 0.10 & 4382.38 & 241.19 & 8\%  & 100\% & 15.63 & 0.54 & 10.35 & 0.08 \\
13 & 4564.46 & 6542.86 & 319.72 & 43\% & 100\% & 18.41 & 1.56 & 10.47 & 0.44 & 5063.10 & 250.39 & 11\% & 100\% & 17.21 & 0.72 & 10.64 & 0.24 \\
14 & 4946.85 & 6446.66 & 271.75 & 30\% & 100\% & 20.46 & 1.26 & 10.51 & 0.15 & 5904.61 & 918.30 & 19\% & 100\% & 18.36 & 1.15 & 10.73 & 0.32 \\
15 & 4967.30 & 9304.35 & 2068.11 & 87\% & 40\%  & 25.28 & 4.35 & 10.51 & 0.39 & 9221.91 & 536.49 & 86\% & 40\%  & 20.97 & 2.71 & 10.23 & 0.15 \\
\midrule
\textbf{AVG} & \textbf{3429,33} & \textbf{4191,14} & \textbf{261,99} & \textbf{16\%} &  & \textbf{17,15} & \textbf{0,84} & \textbf{10,33} & \textbf{0,12} & \textbf{3984,01} & \textbf{177,49} & \textbf{11\%} &  & \textbf{16,36} & \textbf{0,79} & \textbf{10,33} & \textbf{0,09} \\
\bottomrule
\end{tabular}}
\caption{
Results for the TSP comparing performance with and without the CAF method. Each row corresponds to a different instance size defined by the number of vehicles $V$. The table reports: 
\textbf{OF opt}: optimal objective value obtained with Gurobi, 
\textbf{OF Avg}: average objective value across ten runs, 
\textbf{OF Std}: standard deviation of the objective value, 
\textbf{GAP}: average optimality gap with respect to OF opt, 
\textbf{\% Solved}: percentage of instances for which the solver provides a feasible solution within the time limit, 
\textbf{Time Avg}: average total time in seconds (including instance construction and solution time), 
\textbf{Time Std}: standard deviation of the total time,
\textbf{Time Solve Avg}: average solution time in seconds, 
\textbf{Time Solve Std}: standard deviation of the solution time.
The last row (\textbf{AVG}) reports the average values for each column across all tested instances.
}
\label{tab:performance-comparison1}
\vspace{-0.3cm}
\end{table*}

The computational results highlight the significant impact of the CAF preprocessing on the performance of the proposed approach. The experiments consider instances up to $V=15$, as larger instances cannot be reliably solved within the imposed time limit. 
{For $V=15$, the solver successfully solves $40\%$ of the runs within the imposed time limit, while in the remaining $60\%$ of the runs no feasible solution is found within the time limit. This highlights the increasing difficulty of the problem and the role of the CAF preprocessing in improving model tractability.}
In particular, for smaller instances ($V \leq 11$), both formulations consistently reach the optimal solution in all runs. However, as the instance size increases, noticeable differences between the two approaches emerge. In particular, the CAF-based formulation generally produces solutions with a lower optimality gap compared to the original model solved with the classical solver. For example, when $V=12$, the gap decreases from $14\%$ to $8\%$, while for $V=13$ it is reduced from $43\%$ to $11\%$. Overall, when considering all tested instances, the use of CAF leads to an average reduction of approximately $31\%$ in the optimality gap. These results indicate that the preprocessing helps the solver focus on more promising arcs, thereby improving the quality of the obtained solutions.
Moreover, the CAF preprocessing also improves the computational efficiency of the proposed approach. For the largest instance ($V=15$), { it is important to note that only $40\%$ of the runs yield feasible solutions within the imposed time limit. Consequently, the reported optimality gaps for this instance size are based on a limited subset of runs and should be interpreted with caution. Nevertheless, even in this case,} the average solution time decreases from $25,28$ seconds to $20,97$ seconds, corresponding to a reduction of approximately $17\%$. Considering all tested instances, the average computational time decreases from $17,15$ seconds to $16,36$ seconds, yielding an overall reduction of about $5\%$. 
These results show that limiting the number of candidate arcs leads to a more compact optimization model, which can be solved more efficiently by the solver.
\vspace{-0.3cm}
\section{Conclusions}\label{sec:conclusions}
This paper addressed the { symmetric} TSP by proposing a preprocessing technique aimed at reducing the size of the optimization model while preserving feasibility. The proposed CAF procedure restricts the set of candidate arcs by retaining only the lowest-cost neighbors for each vertex. From a theoretical perspective, the preprocessing guarantees that the reduced graph maintains sufficient connectivity to admit at least one Hamiltonian cycle. Consequently, the feasibility of the optimization problem is preserved even after significantly reducing the number of candidate arcs.
Computational experiments show that CAF reduces the number of decision variables by approximately $30\%$ compared to the original formulation. When applied within a classical ILP framework, the preprocessing {preserves feasibility and yields high-quality solutions} while improving computational efficiency, demonstrating that the approach is beneficial also for traditional optimization methods.
The advantages of CAF become even more relevant in the context of quantum optimization, where the number of variables directly affects the tractability of the model. Experimental results obtained with the hybrid quantum solver show that CAF reduces the average optimality gap by approximately $29\%$ while also decreasing computational time.
Overall, the results demonstrate that the proposed preprocessing strategy effectively reduces the {size of the TSP formulation, leading to more tractable optimization models,}
 while maintaining high-quality solutions in both classical and quantum optimization settings.  This highlights the potential of CAF as a practical tool for improving the scalability of combinatorial optimization models, particularly in the presence of current hardware limitations in quantum computing.

Future research may further explore strategies to enhance the scalability of the proposed approach, particularly in the context of large-scale TSP instances and quantum optimization frameworks. Possible directions include the investigation of alternative model reduction techniques and iterative solution procedures that dynamically refine the optimization model during the solution process. Such developments could further improve computational efficiency and facilitate the integration of classical optimization methods with emerging hybrid quantum–classical solvers.
\vspace{-0.35cm}
\bibliographystyle{IEEEtran}
\bibliography{bibliography}

@article{cellini2024qal,
  title={Qal-bp: an augmented lagrangian quantum approach for bin packing},
  author={Cellini, Lorenzo and Macaluso, Antonio and Lombardi, Michele},
  journal={Scientific Reports},
  volume={14},
  number={1},
  pages={5142},
  year={2024},
  publisher={Nature Publishing Group UK London}
}

@article{ciacco2025steiner2,
  title={Steiner Traveling Salesman Problem with Time Windows and Pickup-Delivery: integrating classical and quantum optimization},
  author={Ciacco, Alessia and Guerriero, Francesca and Osaba, Eneko},
  journal={arXiv preprint arXiv:2508.17896},
  year={2025}
}

@inproceedings{li2023quantum,
  title={Quantum Computing Approaches to Optimize Employee Scheduling in Multi-task Call Centers},
  author={Li, Cheng and Liu, Zhaoyang and Song, Yu and Liu, Haojie and Liu, Hanlin and Liu, Xiaodong},
  booktitle={The International Conference on Smart Manufacturing, Industrial \& Logistics Engineering (SMILE)},
  pages={3--9},
  year={2023},
  organization={Springer}
}

@article{holliday2025advanced,
  title={Advanced Quantum Annealing Approach to Vehicle Routing Problems with Time Windows},
  author={Holliday, James B and Blount, Darren and Osaba, Eneko and Luu, Khoa},
  journal={arXiv preprint arXiv:2503.24285},
  year={2025}
}

@inproceedings{de2022hybrid,
  title={Hybrid quantum-classical heuristic for the bin packing problem},
  author={de Andoin, Mikel Garcia and Osaba, Eneko and Oregi, Izaskun and Villar-Rodriguez, Esther and Sanz, Mikel},
  booktitle={Proceedings of the Genetic and Evolutionary Computation Conference Companion},
  pages={2214--2222},
  year={2022}
}

@inproceedings{malviya2023logistics,
  title={Logistics network optimization using quantum annealing},
  author={Malviya, Gajendra and AkashNarayanan, B and Seshadri, Janani},
  booktitle={International Conference on Emerging Trends and Technologies on Intelligent Systems},
  pages={401--413},
  year={2023},
  organization={Springer}
}

@article{ciacco2026facility,
  title={Quantum annealing for the two-level facility location problem},
  author={Ciacco, Alessia and Guerriero, Francesca and Saccomanno, Francesco Paolo},
  journal={Future Generation Computer Systems},
  volume={174},
  pages={107961},
  year={2026},
  publisher={Elsevier}
}

@inproceedings{osaba2025quantum,
  title={Quantum-Assisted Automatic Path-Planning for Robotic Quality Inspection in Industry 4.0},
  author={Osaba, Eneko and Garrote, Estibaliz and Miranda-Rodriguez, Pablo and Ciacco, Alessia and Cabanes, Itziar and Mancisidor, Aitziber},
  booktitle={2025 IEEE International Conference on Quantum Computing and Engineering (QCE)},
  volume={2},
  pages={388--389},
  year={2025},
  organization={IEEE}
}

@inproceedings{ciacco2025steiner,
  title={Steiner Traveling Salesman Problem with Quantum Annealing},
  author={Ciacco, Alessia and Guerriero, Francesca and Osaba, Eneko},
  booktitle={Proceedings of the Genetic and Evolutionary Computation Conference Companion},
  pages={2412--2418},
  year={2025}
}

@article{jain2021tsp,
  author  = {Jain, Siddharth},
  title   = {Solving the Traveling Salesman Problem on the D-Wave Quantum Computer},
  journal = {Frontiers in Physics},
  volume  = {9},
  pages   = {760783},
  year    = {2021},
  doi     = {10.3389/fphy.2021.760783}
}

@article{warren2020tsp,
  author  = {Warren, Richard H.},
  title   = {Solving Combinatorial Problems by Two D-Wave Hybrid Solvers: A Case Study of Traveling Salesman Problems in the TSP Library},
  journal = {arXiv preprint arXiv:2106.05948},
  year    = {2021},
  url     = {https://arxiv.org/abs/2106.05948}
}

@article{bochkarev2026quantum,
  author  = {Bochkarev, Alexey and Heese, Raoul and J{\"a}ger, Sven and Schiewe, Philine and Sch{\"o}bel, Anita},
  title   = {Quantum Computing for Discrete Optimization: A Highlight of Three Technologies},
  journal = {European Journal of Operational Research},
  volume  = {329},
  pages   = {747--766},
  year    = {2026},
  doi     = {10.1016/j.ejor.2025.07.063}
}

@article{rajak2023quantum,
  title={Quantum annealing: An overview},
  author={Rajak, Atanu and Suzuki, Sei and Dutta, Amit and Chakrabarti, Bikas K},
  journal={Philosophical Transactions of the Royal Society A},
  volume={381},
  number={2241},
  pages={20210417},
  year={2023},
  publisher={The Royal Society}
}

@article{hauke2020perspectives,
  title={Perspectives of quantum annealing: Methods and implementations},
  author={Hauke, Philipp and Katzgraber, Helmut G and Lechner, Wolfgang and Nishimori, Hidetoshi and Oliver, William D},
  journal={Reports on Progress in Physics},
  volume={83},
  number={5},
  pages={054401},
  year={2020},
  publisher={IOP Publishing}
}

@article{benson2023cqm,
  title={A CQM-based approach to solving a combinatorial problem with applications in drug design},
  author={Benson, B Maurice and Ingman, Victoria M and Agarwal, Abhay and Keinan, Shahar},
  journal={arXiv preprint arXiv:2303.15419},
  year={2023}
}

@incollection{karp1972reducibility,
  title={Reducibility among combinatorial problems},
  author={Karp, Richard M},
  booktitle={50 Years of Integer Programming 1958-2008: from the Early Years to the State-of-the-Art},
  pages={219--241},
  year={2009},
  publisher={Springer}
}

@article{reinelt1991tsplib,
  title={TSPLIB—A traveling salesman problem library},
  author={Reinelt, Gerhard},
  journal={ORSA journal on computing},
  volume={3},
  number={4},
  pages={376--384},
  year={1991},
  publisher={Informs}
}

@incollection{applegate2011traveling,
  title={The traveling salesman problem: a computational study},
  author={Applegate, David L and Bixby, Robert E and Chv{\'a}tal, Va{\v{s}}ek and Cook, William J},
  booktitle={The traveling salesman problem},
  year={2011},
  publisher={Princeton university press}
}

@book{reinelt2001traveling,
  title={The traveling salesman: computational solutions for TSP applications},
  author={Reinelt, Gerhard},
  year={2001},
  publisher={Springer}
}

@article{dirac1952some,
  title={Some theorems on abstract graphs},
  author={Dirac, Gabriel Andrew},
  journal={Proceedings of the London Mathematical Society},
  volume={3},
  number={1},
  pages={69--81},
  year={1952},
  publisher={Oxford University Press}
}

@article{osaba2024solving,
  title={Solving a real-world package delivery routing problem using quantum annealers},
  author={Osaba, Eneko and Villar-Rodriguez, Esther and Asla, Ant{\'o}n},
  journal={Scientific Reports},
  volume={14},
  number={1},
  pages={24791},
  year={2024},
  publisher={Nature Publishing Group UK London}
}

@inproceedings{ciacco2025Educational,
  title={Quantum annealing for staff scheduling in educational environments},
  author={Ciacco, Alessia and Guerriero, Francesca and Osaba, Eneko},
  booktitle={2026 International Conference on Quantum Communications, Networking, and Computing (QCNC)},
  pages={630--637},
  year={2026},
  organization={IEEE}
}

@article{osaba2025d,
  title={D-wave’s nonlinear-program hybrid solver: Description and performance analysis},
  author={Osaba, Eneko and Miranda-Rodriguez, Pablo},
  journal={IEEE Access},
  year={2025},
  publisher={IEEE}
}

@article{ciacco2025review,
  title={Review of quantum algorithms for medicine, finance and logistics},
  author={Ciacco, Alessia and Guerriero, Francesca and Macrina, Giusy},
  journal={Soft Computing},
  volume={29},
  number={4},
  pages={2129--2170},
  year={2025},
  publisher={Springer}
}
\end{document}